\documentclass[journal]{IEEEtran}

\usepackage[utf8]{inputenc}
\usepackage[english]{babel}

\usepackage{cite}
\ifCLASSINFOpdf
   \usepackage[pdftex]{graphicx}
   \graphicspath{{Images/}}
    \usepackage{graphicx}
\else
\fi

\usepackage{amsthm}
\usepackage{mathtools}
\usepackage{cancel}
\usepackage{xcolor}
\usepackage{relsize}
\usepackage{amssymb}
\usepackage{blindtext}
\usepackage{amsmath}
\usepackage{verbatim}
\usepackage{epsfig}
\usepackage[usenames,dvipsnames]{pstricks}
\usepackage{psfrag}
\usepackage{booktabs}

\usepackage{bm}
\usepackage{multirow}
\usepackage{multicol}
\usepackage{mathdots}
\usepackage{pst-grad} 
\usepackage{pst-plot} 

\newtheorem{theorem}{Theorem}[section]
\newtheorem{lemma}{Lemma}[section]

\newtheorem{definition}{Definition}
\newtheorem{assumption}{Assumption}

\newtheorem{remark}{Remark}

\usepackage{array}
\usepackage[ruled,vlined]{algorithm2e}

\ifCLASSOPTIONcompsoc
  \usepackage[caption=false,font=normalsize,labelfont=sf,textfont=sf]{subfig}
\else
  \usepackage[caption=false,font=footnotesize]{subfig}
\fi

\usepackage{stfloats}

\usepackage{lipsum}
\usepackage{color}
\usepackage{cite}

\newcommand{\G}{\mathcal{G}}
\newcommand{\V}{\mathcal{V}}
\newcommand{\E}{\mathcal{E}}

\begin{document}

\title{Distributed Power Apportioning with Early Dispatch for Ancillary Services \\in Renewable Grids}


\author{Sourav~Patel,~\IEEEmembership{Student Member,~IEEE,}
        Blake~Lundstrom,~\IEEEmembership{Senior Member,~IEEE,}
        Govind~Saraswat,~\IEEEmembership{Member,~IEEE,}
        Murti~V.~Salapaka,~\IEEEmembership{ Fellow,~IEEE}
        \vspace{-1cm}
\thanks{This work was authored in part by the National Renewable Energy Laboratory, operated by Alliance for Sustainable Energy, LLC, for the U.S. Department of Energy (DOE) under Contract No. DE-AC36-08GO28308. Funding provided by the Advanced Research Projects Agency-Energy (ARPA-E) under grant no. DE-AR0000701. The views expressed in the article do not necessarily represent the views of the DOE or the U.S. Government. The U.S. Government retains and the publisher, by accepting the article for publication, acknowledges that the U.S. Government retains a nonexclusive, paid-up, irrevocable, worldwide license to publish or reproduce the published form of this work, or allow others to do so, for U.S. Government purposes.}%

\thanks{S. Patel and M. Salapaka are with the Department of Electrical and Computer Engineering, University of Minnesota, Minneapolis, 55455 MN, USA~
 e-mail: patel292@umn.edu, murtis@umn.edu}
\thanks{B. Lundstrom and G. Saraswat are with the National Renewable Energy Laboratory, Golden, 80401 CO, USA %
e-mail: blake.lundstrom@nrel.gov,govind.saraswat@nrel.gov}}
\maketitle

\begin{abstract}
This article develops a distributed framework for coordinating distributed energy resources (DERs) in a power network to provide secondary frequency response (SFR) as an ancillary service to the bulk power system. A distributed finite-time protocol-based solution is adopted that allows each DER in the network to determine power reference commands. The distributed protocol respects information exchange constraints posed by a communication network layer while being robust to delays in the communications channels. The proposed framework enables coordinated response and control of the aggregated DERs by apportioning the share of generation that each DER needs to provide towards meeting any specified global SFR command while allowing for adjustments due to variability in generation and demand in order to prioritize renewable energy sources in the network. A novel early dispatch mechanism with brown start is synthesized to achieve initial DER response to changing SFR commands that is faster than state-of-the-art distributed approaches. The proposed power apportioning protocol is validated using an end-to-end power hardware-in-the-loop configuration at a distribution system scale with 40+ physical hardware DERs, underlying 7-MW power system model, a 250-DER communication topology with physical and simulated distributed controller nodes, varied communication protocols, and an underlying real-world power system model. Experimental results  demonstrate the efficacy of the proposed method toward distributed coordination of hundreds of DERs for providing fast response at SFR timescales.

\end{abstract}


\IEEEpeerreviewmaketitle
\vspace{-0.5cm}

\section{Introduction}
\IEEEPARstart{D}{}ispatch methodologies for deploying ancillary services for reliable operation of the modern grid are changing significantly as conventional baseload generating units are being replaced by a large number of smaller DERs scattered throughout the network (termed as \textit{decentralization)}. The addition of renewable energy sources (RES) such as, wind and photovoltaic (PV) generation, with increased  generation variability necessitates increased flexibility of dispatch operations \cite{hedayati2017reserve, ela2012studying}. 
 As a result, providing essential reliability services, such as, secondary frequency response (SFR), which corrects for imbalances between total generation and load in the system to restore frequency to its nominal value needs to emphasize DERs and RESs as a focus. An increase in power electronics-interfaced DERs, energy storage systems (ESS) and flexible loads interacting with the grid have shown promise in providing SFR, (for example,  the Electric Reliability Council of Texas's (ERCOT's) responsive reserve services (RRS)), thereby demonstrating the viability of wide-scale adoption of DERs for such applications \cite{miso,konidena2019ferc}. 
A tighter integration of the distribution system operators (DSO) with aggregators, independent power producers and a large number of residential/commercial units with DER assets (`prosumers') as emerging potential participants is exigent.

To achieve distribution-level SFR with a decentralized framework, a large number of DERs must be coordinated on a fast timescale. Earlier works in the literature, such as \cite{xi2018power}, focused on implementing a centralized control approach for coordinating DERs, wherein a secondary centralized controller at the distribution level collects states of DERs via a communication network and sends dispatch commands to local actuators. Such centralized approaches lack flexibility and scalability in providing SFR, and they require expensive high-performance computing and high-speed communication networks to meet SFR requirements satisfactorily. Also, centralized control approaches have added challenges of reduced resiliency. In order to mitigate these challenges, distributed control approaches are proposed with DERs as a multi-agent system (MAS) \cite{Rassa_2019}.
\par Advantages of distributed approaches include coordination using only local computations, plug-and-play capability, and resiliency to node failures. A gather-broadcast method was presented in \cite{dorfler2017gather}, and \cite{zhao2015distributed} developed a distributed average integral method; however, these methods are often sensitive to gain coefficients and can result in slow convergence to dispatch outputs. Reference \cite{megel2017distributed} presented a distributed SFR approach but did not consider RESs and relies on heuristics to achieve dispatch requests. Distributed \textit{consensus}-based algorithms form a primary approach for many applications where distributed decision making is needed; in this article, achieving consensus between agents  on decision variables, forms a main thrust for realizing fast distributed SFR. Major challenges faced by consensus-based distributed approaches include: (i) inherent delays in communication channels that has significant impact on accuracy \cite{async_Antsaklis} and (ii) the asymptotic nature of convergence of the algorithms entails in principle that the result is known only with a infinite horizon. Here, for coordination for SFR, the result of the distributed algorithm will determine the power dispatch command of each individual DER. Thus there is a need for a stopping criterion at each DER, which can be utilized to terminate the consensus algorithm and for determining its dispatch decision. 
\par We now briefly describe the analytical frameworks on consensus based distributed approaches without any specific emphasis on the problem of distributed power apportioning.  Reference \cite{yang2019distributed} utilized the ratio-consensus algorithm to optimally coordinate DERs over time-varying directed communication networks for providing slower timescale tertiary support to the grid. However, \cite{yang2019distributed} does not address communication delays that plague any practical implementation. Moreover, as alluded to earlier, a challenge with distributed algorithms that exchange information with neighbors and update their state multiple times over many iterations is the need to ascertain when to stop iterating and use the decision parameter for a subsequent action, such as determining how much power a DER needs to dispatch. Here, if the detection of convergence can be achieved (within a pre-specified tolerance value) by the nodes distributedly, algorithm run-times longer than necessary can be avoided, making it possible to employ low computational footprint, low cost devices.
To circumvent the issue of asymptotic convergence, \cite{sundaram2007finite,tran2013distributed} proposed finite-time algorithms to compute the consensus value using network observability; 
however, here, limitations of high computational footprint and large storage requirements at each node, render the approach unsuitable for applications, such as fast distributed SFR, where coordination with a fast response time is required.
Reference \cite{cady2015finite} proposed a distributed finite time termination of ratio consensus, which built on \cite{dominguez2010coordination} and \cite{yadav2007distributed}, for frequency regulation in a network of islanded ac microgrids, but it did not consider communication delays in its formulation. Distributed finite-time termination of ratio consensus in the presence of bounded delays are presented in our earlier work \cite{prakash2019distributed}.

\par All of the analytical works reported above are not instantiated to  large scale coordination of DERs; thus an effective framework for  distributed aggregation based SFR is currently absent.  We remark that in  \cite{dominguez2010coordination},  distributed apportioning was achieved for a few DERs. Here, the algorithms used did not address the issues of non-ideal nature of communication, nor any guarantees on the finite-time distributed stopping criterion are utilized. Without addressing these issues, the consensus based strategies remain inapplicable for practical sized  distributed DER aggregation goals.  Indeed, to the best of the authors' knowledge, no prior work or framework has demonstrated a provably guaranteed distributed method, with instantiation of a large scale coordination to achieve DER based SFR.  This article develops a scalable DSO-centric  framework towards coordinating large numbers (1000+) of DERs to provide support to system operators via aggregators in dispatching SFR. The framework meets the SFR ancillary demand (global objective) of the system operator by aggregating \textit{distribution-level} DERs in a distributed manner while respecting local capacity constraints of each DER. Coordinated response and control of aggregated DERs are achieved at SFR timescales: initial response times of less than 5 seconds and a ramp to response set point within 1 minute are demonstrated. %
Moreover, the distributed power apportioning framework comes with guarantees on reaching desired and feasible dispatch decisions in finite-time, even when the communication suffers from uncertainties such as delays.
The developed framework's efficacy is demonstrated using low-cost \textit{Raspberry Pi} (Rpi) devices (with local communication and computational intelligence capabilities) interfacing with an underlying power controller layer. The framework here is robust to the presence of bounded delays in the communication channels.
 Further, early dispatch and brown-start mechanisms to enable participating DERs to respond to an SFR signal at timescales faster than the state-of-the-art distributed approaches are developed. %
We noe summarize the major contributions of this article:\\
1) A distributed framework for coordinated power apportioning in the presence of communication delays that preserves the privacy of private values (capacity, output power) of the participating DERs.\\
2) An extension of distributed stopping criteria \cite{yadav2007distributed} where each DER in the network can detect convergence within a tolerance independently 
to formulate a novel \textit{early dispatch} mechanism. This allows the network to achieve SFR faster than state-of-the-art distributed approaches (initial response time of less than 5 s and ramp time of 1 minute).\\
3) A brown-start approach whereby DERs, already participating in SFR, smoothly transition to new states in order to meet new system operator commands in the presence of changes in generation capacities, is developed. This contributions extends authors' work in \cite{patel2017distributed}.\\
4) A large scale validation employing 40+ physical hardware DERs and 250 nodes on a real-world distribution system model.
To the best of our knowledge, distributed DER coordination and control at such a scale is not demonstrated earlier with a large-scale validation. 

\section{System Description}\label{sec:sys_desc}
This section describes the system under study to develop the resource apportioning problem. A detailed description of graph theory and linear algebra notions that are used in the subsequent development are available in \cite{diestel2000graph}.
We consider DERs (such as a PV array and battery ESS), each interfaced with DC-AC inverter connected to local ac loads and a grid connection and to other units through a point of common coupling (PCC) in a microgrid or to the grid through an aggregator (see Fig. \ref{fig:PHILsystem}). 
\vspace{-0.3cm}
\subsection{Communication Network of DER Units}

 In order to facilitate exchange of information to arrive at viable power commands for meeting ancillary demand services of the grid, we consider a network of DER units as a multi-agent system (MAS)  with  agents/nodes  interacting  with  their  neighbors,  over a  communication  infrastructure (\textit{Communication Layer} in Fig. \ref{fig:PHILsystem}). We consider a graph $\G=\{\V,\E\}$, where vertices $\V =\{1,2,\ldots,N\}$ denote the DER nodes in the network, and $\E \subseteq \V \times \V$ is the set of edges representing the communication topology overlay on the underlying power system infrastructure. DER units communicate with each other as dictated by the network topology where each node commuicates with its neighbors using multiple allowable communication modalities (wired, wireless, hybrid). Communication can be bidirectional or directional. In the network representation (see Fig.~\ref{fig:PHILsystem}), a bidirectional channel is represented using a bi-directed edge and a directed channel with a directed edge. The underlying communication (implemented here using Rpi devices) is uncertain suffering from delays.
  \begin{assumption}\label{assmp:connected}
  $\G=\{\V,\E\}$ is connected.
  \end{assumption}
 \begin{assumption}\label{assmp:delayBound}
For any node pair $i,j \in \V$ and $(i,j) \in \E$, the delay on the edge from node $j$ to node $i$, denoted as $\tau_{ij}$, satisfies: $\tau_{ij}\leq \bar{\tau}<\infty$, where $\bar{\tau} >0 $ .
 \end{assumption}
 
 Assumption \ref{assmp:delayBound} reflects the \textit{partial synchrony} condition for a real-world communication network in distributed computing. This framework also adheres to communication protocols that guarantee no packet loss, such as  Transmission Control Protocol-based websockets protocols. 
 
 We consider, an \textit{aggregator} to be an entity interfacing with the DSO on one end and prosumers on the other. The aggregator is responsible for accumulating DERs available at the distribution level to provide a grid ancillary power command, $\rho_d$, to the network of DERs. Here, the aggregator can communicate the command $\rho_d$ only to nodes in its communication neighborhood; here the aggregator is assumed to have $l\geq 1$ neighbors. The demand signal considered here is the SFR signal where the following specifications have to be met (i) initial response time of less than 5 s, (ii) ramp response time to set point within 1 minute, (iii) maintaining the desired SFR command for a maximum time period of at least 
 30 minutes or as required by the system operator. DERs have the capability to rapidly adjust their dispatch set points to output the SFR commands. 
 

\begin{figure*}[t]
	\centering
\includegraphics[scale=0.105, trim={3cm 7cm 0cm 0cm},clip]{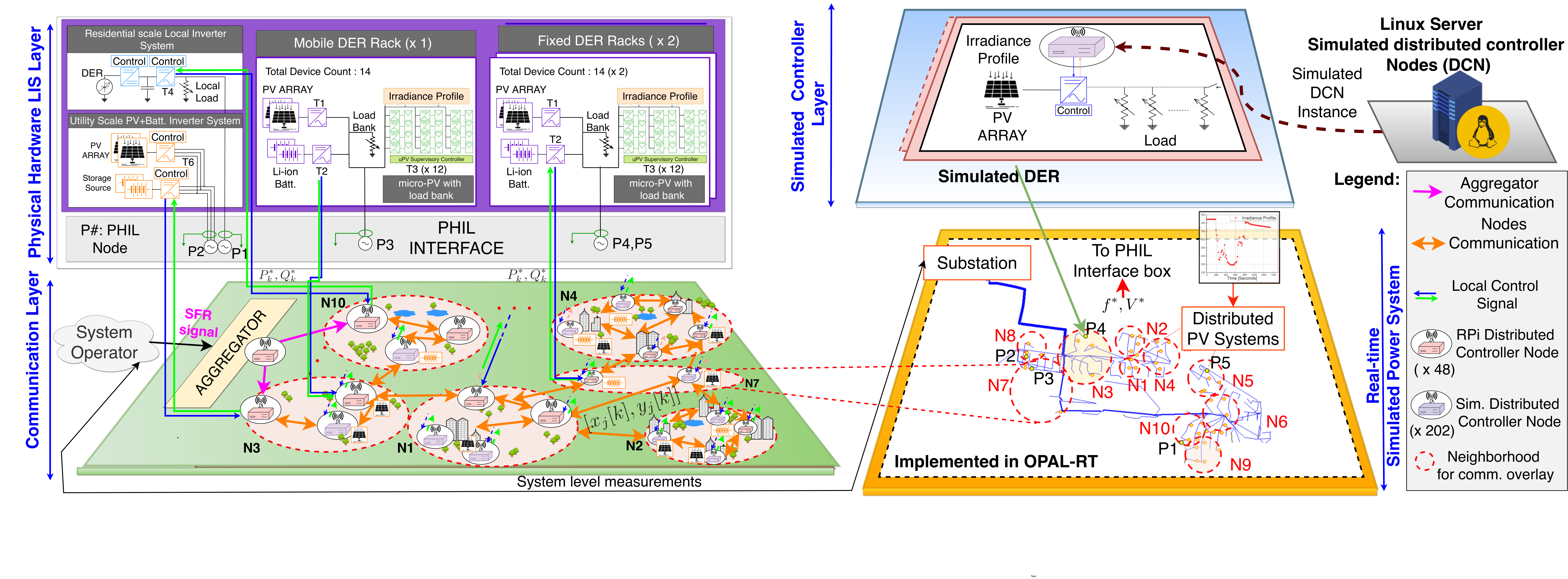}
\caption{\small{ Schematic diagram of the system under study, consisting of  residential, utility-scale DER units; distribution system model with interfaces to controllers with simulated DER units; medium-high power hardware-in-the-loop and aggregator interfacing the distribution system to the rest of the grid.} }
\label{fig:PHILsystem}
\end{figure*}
\subsection{Problem Formulation}
An aggregator, upon receiving a system operator's SFR signal (can be commanded manually or automatically with less than $ 6$ s refresh rates) at time instants $t_0< t_1< t_2< \ldots$, sends a power command, $\rho_d(t_m)$, at time instant $t_m$, where $m \in \{0,1,2, \ldots\}$, to one or more aggregator-facing DER in the network.  The $N$ participating DERs with a given graph topology, $\G=\{\V,\E\}$, need to collectively meet the aggregator command, $\rho_d(t_m)$, while communicating only to their neighboring (neighbors determined by the communication network layer) DER units and respecting their individual resource (generation) constraints.
Let $\pi_i^{max}(t_m)$ and $\pi_i^{min}(t_m)$ be the maximum and minimum generation capacities of the DERs at the $i^{th}$ DER, respectively, at instant $t_m$. 
Let $\pi_i^{*}(t_m)$ denote the steady-state reference power command to be supplied by the $i^{th}$ DER in response to an aggregator command, $\rho_d(t_m)$. The resource apportioning problem can be formulated mathematically as: for any time instant $t_m$ determine;
\begin{align}\label{eq:resApportioning}
    &\Big\{\pi_i^*(t_m)\Big\}_{i=1}^N \nonumber\\
    \mbox{such that } &\sum_{i\in \V} \pi_{i}^{*}(t_m) = \rho_d(t_m) \\
       &\pi_i^{min}(t_m) \leq \pi_i^{*}(t_m) \leq \pi_i^{max}(t_m) \, \; \text{for all } i \in \V \nonumber. 
\end{align}
 Furthermore, the viability of the resource apportioning problem holds when  $\sum_{i=1}^N \pi_i^{min}(t)\leq \rho_{d}(t)\leq \sum_{i=1}^N \pi_i^{max}(t)$; otherwise for $\rho_d(t) \ge \sum_{i=1}^{N} \pi_i^{max}(t)$, we set $\pi_i^*(t) = \pi_i^{max}(t)$ for all $i \in \V$.

We remark that the resource apportioning problem (\ref{eq:resApportioning}) needs to be solved faster than the time interval between $t_{m-1}$ and $t_m$  when a new dispatch command is placed. In real-world power networks the maximum time to respond required by the system operator is in the order of minutes \cite{sandia2012survey}. The approximate solution to (\ref{eq:resApportioning}) is found based on a finite-time termination criteria using a user-selected tolerance parameter which provides some flexibility wherein faster convergence results with a larger tolerance.
\vspace{-0.35cm}
\section {Apportioning Using Distributed Averaging}\label{sec:apportioning}
The resource apportioning problem with constraints can be solved using average consensus protocols. We first summarize the distributed averaging protocol \cite{yadav2007distributed, hadjicostis2014average,prakash2019distributed}, and its extension toward solving the resource apportioning problem for the asymptotic case.
\vspace{-0.4cm}
\subsection{Distributed Averaging Protocol}
\begin{assumption}\label{assmp:Pmatrix}
 Let $p_{ij}$ denote the weight on information coming from node $j$ to node $i$. The weight matrix $P(i,j)=p_{ij} \geq 0$ associated with $\G$ is primitive (if it is irreducible and has only one eigenvalue of maximum modulus) and column stochastic (all elements of every column of the matrix sum to one).
\end{assumption}
\begin{assumption}\label{assmp:selfNode}
 Any node $i \in \V$ in $\G=(\V,\E)$ has access to its own value at any instant $k$ without any delay.
\end{assumption}
\begin{definition}(In-neighbor Set and Out-neighbor Set)
In a graph $\G:=(\V,\E)$, the in-neighbor set, $\mathcal{N}_j^{-}$,  and out-neighbor set, $\mathcal{N}_j^+$, of node $j \in \V$, are given as $\mathcal{N}_j^{-}:=\{i| (j,i)$ $ \in \E, i \neq j\}$ and $\mathcal{N}_j^{+}:=\{i| (j,i)$ $ \in \E, i \neq j\}$, respectively. The cardinality of the  out-neighbor set of a node $i \in \V$ is called the out-degree of the node denoted by $D_{i}^{+}$. 
\end{definition}
Consider the following update iterations for states $x$ and $y$ maintained by all nodes in the network:
\begin{align}
& \textstyle  x_{i}(k+1)=p_{ii}x_{i}(k)+\sum_{j\in\mathit{N_{i}}^{-}}p_{ij}x_{j}(k-\tau_{ij}), \label{eq: Num} \\ %
&\textstyle  y_{i}(k+1)=p_{ii}y_{i}(k)+\sum_{j\in\mathit{N_{i}}^{-}}p_{ij}y_{j}(k-\tau_{ij}),  \label{eq: Den}
\end{align}
where $\tau_{ij}$ is the delay in receiving data from node $j$ to node $i.$

\begin{theorem} \label{thm:avg_con}
(\textit{Ratio Consensus} \cite{hadjicostis2014average}) Suppose Assumptions \ref{assmp:connected}\textendash\ref{assmp:selfNode} hold. Let the initial conditions for the numerator states be  given as $x(0)=[x_{1}(0)\ x_{2}(0) \ldots x_{N}(0)]^{T}$ and $y(0)=[y_{1}(0)\ y_{2}(0) \ldots y_{N}(0)]^{T}$. Then the ratio $\frac{x_{i}(k)}{y_i(k)}$ asymptotically converges to $\frac{{\sum_{i=1}^{N}}x_{i}(0)}{{\sum_{i=1}^{N}}y_{i}(0)}$
for all $i=1,...,N$. 
\end{theorem}
\begin{remark}
 The ratio $x_i(k)/y_i(k)$ is only well-defined when $y_i(k)>0$, which is guaranteed when $y_i(0)>0$ for all $i \in \V$.
\end{remark}
\begin{remark} (Distributed Synthesis)
The weight matrix being column stochastic enables the weights $p_{ij}$ to be chosen in a purely distributed manner. A simple scheme is that $i^{th}$ node sets the weights $p_{ji} = \frac{1}{D_i^{+} + 1}$ for all $j \in \{\mathcal{N}_{i}^{+} \cup {i}\}$  and communicates to node $j$, $\sigma_{ji}(k)=p_{ji}x_i(k)$. Each node $j$ executes $x_j(k)=p_{jj}(k)+\sum_{i} \sigma_{ji}(k-\tau_{ji})$ thus realizing (\ref{eq: Num}) and (\ref{eq: Den}).
\end{remark}
\subsection{Power Apportioning Protocol}
Let the DER communication network be represented by a graph $\G(\V,\E)$. Let $\mathcal{N}_d$ denote the set, $l = \vert \mathcal{N}_d \vert$, of nodes directly communicating with the aggregator, referred to as \textit{command circulating nodes}. Upon receiving a DSO signal for the dispatch of DER as part of the SFR at any time instant $t_m$, the aggregator communicates the ancillary service request, $\rho_d(t_m)$, to the network of DERs by communicating to the command circulating nodes in $\mathcal{N}_d$. At each time instant $t_m$, each node $i$ initializes two states $[r_i(0),s_i(0)]^{T}$ such that:
\begin{align}
    r_i(0) &=
\begin{cases}
    \dfrac{\textstyle  \rho_d(t_m)}{l} - \textstyle  \pi_i^{min}(t_m), \ \text{if} \ i \in \mathcal{N}_d, \\
 - \textstyle  \pi_i^{min}(t_m), \ \text{if} \ i \not\in \mathcal{N}_d ,\\
\end{cases} \label{eq:initialize1}\\
s_i(0) & = \pi_i^{max}-\pi_i^{min}, ~\text{for all } i \in \V.
\end{align}
Each node executes versions of (\ref{eq: Num}) and (\ref{eq: Den}) as described below: 
\begin{align}
r_i(k+1) &= \textstyle  p_{ii}r_i(k) +  \sum_{j \in N_i^{-}} p_{ij} r_j(k-\tau_{ij}),\label{eq:num1}\\
s_i(k+1) &= \textstyle p_{ii}s_i(k)+ \sum_{j \in N_i^{-}} p_{ij} s_j(k-\tau_{ij}),\label{eq:den1}
\end{align}
where, $\pi_i^{max}(t_m)$ and $\pi_i^{min}(t_m)$ denote the maximum and minimum power capacity of the $i^{th}$ DER at time $t_m$. We remark that, $s_i(0) >0 $ for all $i \in \V$ as $\pi_i^{max}(t_m) > \pi_i^{min}(t_m)$.

\begin{lemma}\label{lem:powerApprProtocol}
 Under Assumptions \ref{assmp:connected} \textendash \ref{assmp:selfNode}, the power apportioning protocol  (\ref{eq:initialize1}) \textendash (\ref{eq:den1}) converges asymptotically, i.e.:
\begin{align}\label{eq:convergence}
     \lim_{k \to \infty} \dfrac{r_i(k)}{s_i(k)} &\rightarrow \dfrac{\textstyle\sum_{i \in \mathcal{N}_d} \Big(\dfrac{\rho_d(t_m)}{l}\Big) - \textstyle\sum\limits_{i=1}^N  \pi_i^{min}(t_m)}{\textstyle\sum_{i=1}^N (\pi_i^{max}(t_m)-\pi_i^{min}(t_m))},
\end{align}
for each node $i \in \V$.
\end{lemma}
\begin{proof}
The result can be obtained directly by noting that for a given time instant $t_m$:\\ $\sum_{i=1}^N r_i(0) =$ $\scriptstyle \sum_{i \in \mathcal{N}_d} \Big(\dfrac{\scriptstyle\rho_d(t_m)}{l}-\pi_i^{min}(t_m)\Big) + \sum_{i \notin \mathcal{N}_d} (- \pi_i^{min}(t_m))$ and $\sum_{i=1}^N s_i(0) = \sum_{i=1}^N (\pi_i^{max}(t_m)-\pi_i^{min}(t_m))$ and applying Theorem \ref{thm:avg_con}.
\end{proof}
\begin{theorem}\label{thm:resouce_allocation}
Let the power reference command for the $i^{th}$ DER due to aggregator command at $t_m$ be defined as, $\textstyle \pi_i^{*} := \pi_i^{min}(t_m) + \lim_{k\rightarrow \infty}\frac{r_i(k)}{s_i(k)}(\pi_i^{max}(t_m) - \pi_i^{min}(t_m)). \ \text{Then} \ \sum_{i=1}^N \pi_i^{*} = \rho_{d}(t_m)$ and $\pi_i^{min}(t_m) \leq \pi_i^* \leq \pi_i^{max}(t_m)$ for all $i\in \V$. 
\end{theorem}
\begin{proof}
See \cite{patel2017distributed}.
\end{proof}
\vspace{-0.8cm}
\subsection{RES Prioritization}
In order to demonstrate the dispatchability of available RES, this article also proposes an RES prioritization scheme. The motivation behind this objective is to dispatch spinning RES-based DERs as SFR before ESS for long-term reserve requirements. This is incorporated in the power apportioning protocol by setting $\pi_i^{min}(t_m) = \pi_i^{max}(t_m) -\epsilon_{RES}$,
where $\epsilon_{RES} > 0$ is a small. 
This enforces all available RES capacities at the initialization of the protocol to be used, allowing for RES prioritization as validated experimentally.\\
Theorem \ref{thm:resouce_allocation} provides a distributed protocol to allocate resources to meet the demand $\rho_d(t_m)$ by a DER network; however, it is clear that this protocol is not amenable to accommodating aggregator commands at subsequent time instants $t_{m+1}, t_{m+2}, \ldots$ as the result in Theorem \ref{thm:resouce_allocation} is asymptotic  (where the desired dispatch, $\pi_i^*$, of the $i^{th}$ DER  is determined only in the limit of iteration $k\rightarrow \infty$).
In order to mitigate this issue, we now propose the formulation of the distributed finite-time termination protocol.

\vspace{-0.5cm}
\section{Distributed Finite-Time Termination of Resource Apportioning}\label{sec:finitetime} 
In this section, we develop an algorithm using the Maximum and minimum consensus protocols for terminating the ratio consensus algorithm in finite-time based on a specified tolerance, $\rho$.
We apply the results in \cite{prakash2019distributed} for distributed finite time termination of the ratio consensus algorithm. The resulting algorithm is presented in Algorithm \ref{alg:algo1}.  

Consider the maximum and minimum value of the ratio of consensus protocols (\ref{eq: Num})\textendash (\ref{eq: Den}) over all nodes within a time horizon $\bar{\tau}$ from any time instant $k$, given as:
        \begin{equation} \label{eq:max_ratio}
    \textstyle M(k):=\underset{r=\{0,1,2,...,\bar{\tau}\}}{\underset{j\in V}{\max}} \frac{x_{j}(k-r)}{y_{j}(k-r)},~ y_j(k-r) \neq 0, j \in V 
\end{equation}
\begin{equation}\label{eq:min_ratio}
     \textstyle m(k):=\underset{r=\{0,1,2,...,\bar{\tau}\}}{\underset{j\in V}{\min}} \frac{x_{j}(k-r)}{y_{j}(k-r)},~ y_j(k-r) \neq 0, j \in V 
\end{equation}

Under assumptions \ref{assmp:connected}\textendash\ref{assmp:selfNode} , 
it can be shown that the global maximum (minimum) in the network is decreasing (increasing). Of particular interest are the maximum and minimum values over  an ``epoch" which is equal to $T(D,\bar{\tau})=D(1+\bar{\tau}) +\bar{\tau}$, which captures the upper bound on number of iterations required for any node in the network to communicate to any other node in the network.


\begin{theorem}{}\cite{prakash2019distributed}\label{thm:monotonicity}
    Consider the initial ratio vector at the beginning of the $k^{th}$ epoch given by  $\frac{x(kT)}{y(kT)} := [\frac{x_1(kT)}{y_1(kT)},\ldots,\frac{x_N(kT)}{y_N(kT)} ]$ such that $\min \frac{x(kT)}{y(kT)} < \max \frac{x(kT)}{y(kT)}$, where, $k=0,1,2,\ldots$. Then:
    \begin{align}
        M((k+1)T)< M(kT), ~ m((k+1)T) > m(kT).
    \end{align}
\end{theorem}
The above result states the the maximum (minimum) over the network at the beginning of a epoch is strictly smaller (larger) than  that at the beginning of a subsequent epoch. Thus the global maximum and minimum sampled at every epoch form strictly monotonic sequences. It can be further shown that these sequences converge to the consensus value, as specified in the next theorem. 
\begin{theorem}\cite{prakash2019distributed}\label{thm:convg}
       $\underset{k\rightarrow \infty}{\lim} M(kT) = \underset{k\rightarrow \infty}{\lim} m(kT) = \frac{\sum_{j=1}^N x_j(0)}{N}$.
\end{theorem}
The theorem above establishes the asymptotic convergence of global maximum and minimum to the ratio of the sum of initial states. Thus, if we can determine the global maximum and minimum at each epoch distributedly, their difference can be used to bound how far each node's ratio state is from the final consensus value.\\
We now introduce the maximum consensus and minimum consensus protocols that determine $M(kT)$ and $m(kT)$  in finite number of iterations.
\vspace{-0.5cm}
\subsection {Maximum and Minimum Consensus Protocols}
 The maximum consensus protocol (MXP) works as follows. Consider at time instant $\ell \bar{\tau}$, the value held by the $i^{th}$ node is $z_i(\ell \bar{\tau})$ with $z_i(0)$ being the initial value at $\ell = 0$. The $i^{th}$ node holds this value for $\bar{\tau}$ units of time; as $\bar{\tau}$ is the upper bound on delay, in $\bar{\tau}$ units of time node $i$ receives $z_j(\ell \bar{\tau})$ for all $j \in \mathcal{N}_i^-$ and thus it can execute the update:
 \begin{subequations}\label{eq: MaxProtocolNodelay}
 \begin{align}
     z_i(\ell \bar{\tau} +\bar{\tau}) &= \underset{j \in N_i^{-}\cup\{i\}}{\max} z_j(\ell \bar{\tau}), \\
     z_i(\ell \bar{\tau} +k) &= z_i (\ell \bar{\tau}), ~ \text{for } k =1,2,\ldots, \bar{\tau} -1.
 \end{align}
 \end{subequations}

Similarly, the minimum consensus protocol (MNP) computes the minimum of the given initial node conditions $w(0):=[w_{1}(0)\ w_{2}(0)....w_{n}(0)]^{T}$ in a distributed manner. Consider at time instant $\ell \bar{\tau}$, the value held by the $i^{th}$ node is $w_i(\ell \bar{\tau})$ with $w_i(0)$ being the initial value at $\ell = 0$. Similar to the MXP protocol, $i^{th}$ node holds this value for $\bar{\tau}$ units of time and in $\bar{\tau}$ units of time node $i$ would have received $w_j(\ell \bar{\tau})$ for all $j \in \mathcal{N}_i^-$ and thus it can execute the update:
\begin{subequations}\label{eq: MinProtocolNodelay}
 \begin{align}
     w_i(\ell \bar{\tau} +\bar{\tau}) &= \underset{j \in N_i^{-}\cup\{i\}}{\min} w_j(\ell \bar{\tau}), \\
     w_i(\ell \bar{\tau} +k) &= w_i (\ell \bar{\tau}), ~ \text{for } k =1,2,\ldots, \bar{\tau} -1.
 \end{align}
 \end{subequations}
\begin{remark}
MXP and MNP converge to the maximum and minimum of the initial conditions, respectively, within  an epoch $T(D,\bar{\tau}):= D(1+\bar{\tau})+\bar{\tau}$ iterations, where $D$ is the upper bound on the diameter of the network \emph{\cite{prakash2019distributed}}. Thus, each node can
 compute the global maximum and minimum distributedly in $D(1+\bar{\tau})+\bar{\tau}$ .
\end{remark}

The MXP and MNP protocols at each node $i \in V$ are re-initialized at every epoch; here, for every $k = \theta (D(1+\bar{\tau})+\bar{\tau}) = \theta T$, where $\theta = 1,2,...$, we initialize $z_i(k) = \frac{r_i(k)}{s_i(k)}$ and $w_i(k) = \frac{r_i(k)}{s_i(k)}$. In an epoch, each node $i$  can determine $M(\theta T) =\max_j z_j(\theta T)$ and $m(\theta T) =\min_j w_j(\theta T).$ Let $M(\theta T)$ and $m(\theta T)$ be the converged values from the MXP and MNP protocols after each epoch $T$. Then, it follows from Theorem \ref{thm:monotonicity} and Theorem \ref{thm:convg} that $M( \theta T) \searrow \frac{\sum_{j=1}^N x_j(0)}{N} $ and $m(\theta T)\nearrow \frac{\sum_{j=1}^N x_j(0)}{N}$.  Thus given any tolerance $\rho$,  there  exists a $\theta_0$ such that if $\theta \geq \theta_0$, $M(\theta T) -m( \theta T) \leq \rho$ where $M(\theta T) -m(T\theta)$ can be determined by every node in finite time. Thus,  the consensus value $\frac{\sum_{j=1}^N x_j(0)}{N}$ can be determined with a tolerance $\rho$ in finite number of iterations. \cite{prakash2019distributed} establishes rigorously that given a threshold $\rho > 0$, Algorithm \ref{alg:algo1} terminates in a finite number of iterations.
\begin{remark}
The only global parameters needed by each node in order to implement the stopping criteria of (\ref{eq:num1}) and (\ref{eq:den1}) are upper bounds on both the maximum delay, $\bar{\tau}$, and diameter of the network, $D$, which can be enforced by design.
 \end{remark}
\vspace{-0.35cm}
\subsection{Early Dispatch Mechanism}\label{subsec:EDmechanism}
In order to meet the requirement of initial response time of $<5$ s for DERs participating in SFR support, we propose an early dispatch mechanism. In this mechanism, we use the ratio value from the most recent epoch to obtain reference power commands for each DER. The ratio value is guaranteed to stay within the global maximum and minimum ratios (computed via the MXP and MNP protocols) for the DER units running the power apportioning protocol due to the monotonicity property of the ratio consensus protocol (Theorem \ref{thm:convg}), the subsequences generated by the MXP-MNP protocols  converging towards the final set point $\pi_i^*(t_m)$ within the tolerance, $\rho$.
Under the early dispatch mechanism the $i^{th}$ DER provides dispatch as:
\begin{align*}
    \displaystyle
    \pi_i^*(t_m,\theta) = \pi_i^{min}(t_m) + \dfrac{r_i(k)}{s_i(k)}(\pi_i^{max}(t_m)-\pi_i^{min}(t_m))
\end{align*}
where, $k = \theta(D(1+\bar{\tau})+\bar{\tau})$ (see Algorithm~\ref{alg:algo1}). Here $\pi_i^*(t_m,\theta) \rightarrow \pi_i^*(t_m)$ for all $i \in \V$. For practical implementation, early dispatch is implemented for $\theta > 3$.
\vspace{-0.3cm}
\subsection{Brown Start: Changes in Power Demand as Input}\label{subsec:BrownStartmechanism}
    In this section, we propose a brown-start mechanism where command circulating nodes require as input, the change in the requested power command from the previous time instant ($\Delta \rho_d(t_m)$) as the aggregator signal where for $m >0$:
    \begin{align}\label{eq:brownRhoD}
        \Delta \rho_d(t_m)= \rho_d(t_m) - \rho_d(t_{m-1}).
    \end{align}
    The advantage of this mechanism is that the algorithm can be reinitialized with the converged state at the previous time instant rather than  by following the initialization given by (\ref{eq:initialize1}) with respect to an aggregator command $\rho_d(t)$ that does not have regard for the operating state of the DER network. The latter approach results in larger deviations in output power between aggregator commands (see Fig. \ref{fig:ratingsDER}(a)). When $\Delta\rho_d$ is small, the algorithm can converge very quickly as all the nodes will have initial values near the consensus state unlike the case where the algorithm is re-initialized with the new command, $\rho_d(t_m)$.
We propose a modified numerator update, $r_i$, which is initialized as, $r_i(0):=$
\begin{align}\label{eq:brownInitialize}
\begin{cases}
   \textstyle \Delta \rho_d(t_m)/l - \pi_i^{min}(t_m) + \pi^{*}_i(t_{m-1}), \ \text{if} \ i \in \mathcal{N}_d, \\
 - \textstyle \pi_i^{min}(t_m)+ \pi^{*}_i(t_{m-1}), \ \text{if} \ i \not\in \mathcal{N}_d 
\end{cases}
\end{align}


With the new initialization of $r_i$ by (\ref{eq:brownInitialize}), instead of (\ref{eq:initialize1}), at a given instant $t_m$ for $m > 0$, we have:
\begin{align}
    \textstyle \sum_{i=1}^N r_i(0) &= \textstyle \sum_{i \in \mathcal{N}_d} \Big(\dfrac{\Delta \rho_d(t_m)}{l}-\pi_i^{min}(t_m))+ \pi^{*}_i(t_{m-1})\Big)  \nonumber\\
    & \quad +\textstyle \sum_{i \notin \mathcal{N}_d} (- \pi_i^{min}(t_m))+ \pi^{*}_i(t_{m-1})) \nonumber
\end{align}
Because the algorithm has converged in previous iteration, we have: $\rho_d(t_{m-1})=\sum_{i \in V}\pi^{*}_i(t_{m-1})$. Thus, $\sum_{i=1}^N r_i(0) = $
\begin{align*} \textstyle
     \textstyle \sum_{i \in \mathcal{N}_d} \Big(\dfrac{\textstyle \rho_d(t_m)}{l}-\textstyle \pi_i^{min}(t_m)\Big) + \textstyle\sum_{i \notin \mathcal{N}_d} (- \pi_i^{min}(t_m)),
\end{align*}
which is the same as when $r_i$ was initialized by (\ref{eq:initialize1}). Thus, Lemma~\ref{lem:powerApprProtocol} still holds. 
\begin{algorithm}[t]
\linespread{0.3}\selectfont
\SetAlgoLined
\small
    \SetKwBlock{Input}{Input:}{}
    \SetKwBlock{Initialize}{Initialize:}{}
    \SetKwBlock{Repeat}{Repeat: At each time instant $t_m$, $(m \in\{0,1,2,\ldots\})$}{}
    \SetKwBlock{Repeatconv}{Repeat: }{}
    \Repeat{ 
        \Input{$\pi_i^{\min}(t_m)$, $\pi_i^{\max}(t_m)$, $\rho, \bar{\tau}, D$} 
        \Initialize{
            \uIf{ $m=0$ \tcp{Black start}}{
                Aggregator input: $\dfrac{\rho_d(t_m)}{l}$ if $i \in \mathcal{N}_d$
                \tcp{Initialize $r_i(0)$ as in (\ref{eq:initialize1})}
            }\Else
             {
                \tcp{Brown start}
                Aggregator input: $\dfrac{\Delta \rho_d(t_m)}{l}$ if $i \in \mathcal{N}_d$
                \tcp{Initialize $r_i(0)$ as in (\ref{eq:brownInitialize})}
            }
            
            $s_i(0)=\pi_i^{\max}(t_m)-\pi_i^{\min}(t_m) $;\\
            $z_i := r_i(0)/s_i(0)$, $w_i := r_i(0)/s_i(0)$;\\
            $k := 0$, $\gamma := 1$, $\theta := 1$;
           }
        \Repeatconv{
        \tcc{ratio consensus updates of node $i$ given by (\ref{eq: Num}), (\ref{eq: Den})}
           
            $r_{i}(k+1) := p_{ii}r_{i}(k)+\sum_{j\epsilon\mathit{N_{i}^-}}p_{ij}r_{j}(k-\tau_{ij})$;\\
            $s_{i}(k+1) := p_{ii}s_{i}(k)+\sum_{j\epsilon\mathit{N_{i}^{-}}}p_{ij}s_{j}(k-\tau_{ij})$;\\
           \If {$k+1=\gamma(\bar{\tau}+1)$} { 
                \tcc{maximum and minimum consensus updates given by (\ref{eq: MaxProtocolNodelay}(a)), (\ref{eq: MinProtocolNodelay}(a)) for node $i$}
                $\displaystyle z_i := \max_{j \in N_i^{-} \cup \{i\}}z_j$, 
                $w_i := \displaystyle \min_{j \in N_i^{-} \cup \{i\}}w_j$;\\
                $\gamma := \gamma+1$
            }
            \textbf{emit:} $r_i(k+1)$, $s_i(k+1)$, $w_i$ and $z_i$\\
            \If {$ k+1=  \theta (D(1+\bar{\tau})+\bar{\tau})$} {
                \uIf {$z_i - w_i < \rho$ } {$r_i^* = r_i(k+1);$\\
                $s_i^* = s_i(k+1);$\\
                \textbf{break} \tcp*{stop $r_i$, $s_i$, $w_i$ and $z_i$ updates}
                }
                \Else {
                $z_{i} := r_{i}(\theta (D(1+\bar{\tau})+\bar{\tau}))/s_{i}(\theta (D(1+\bar{\tau})+\bar{\tau}))$;\\
                $w_{i} :=r_{i}(\theta (D(1+\bar{\tau})+\bar{\tau}))/s_{i}(\theta (D(1+\bar{\tau})+\bar{\tau}))$;\\
                $\theta := \theta+1$;\\
                  $\pi_i^{*}(t_m,\theta) := \pi_i^{min}(t_m) +$\\$\qquad \qquad z_i(\pi_i^{max}(t_m) - \pi_i^{min}(t_m))$;\\ \tcp{Early Dispatch Mechanism}
                }
            }
            $k= k + 1 $;
        }
        $\pi_i^{*}(t_m) := \pi_i^{min}(t_m) + \frac{r_i^*}{s_i^*}(\pi_i^{max}(t_m) - \pi_i^{min}(t_m))$ 
        \tcp{final power reference command for node $i \in V$}
        }
        \caption{Distributed finite-time termination of resource apportioning in the presence of communication delays (at each node $i \in V$)}
        \label{alg:algo1}
\end{algorithm}
\subsection{Insufficient DER Capacity for SFR Dispatches}
Typically in electricity markets, enough generation reserves and controllable load resources are committed such that North American Electric Reliability Corporation performance standards are always met;
however, since in the proposed distributed framework the aggregator/DSO does not have access to real-time capacities of participating DERs, it is imperative that during insufficient capacity periods, i.e., when $\sum_{i=1}^N \pi_i^{max}(t_m) < \rho_d(t_m)$, the maximum generation capacity be commanded in order to minimize the error between the commanded and dispatched output power. Algorithm \ref{alg:algo1} in such a scenario will command DER units to dispatch output power, $\pi^*(t_m) = \pi_i^{max}(t_m)$.
This is further validated in Test Case II.
\vspace{-0.3cm}
\section{Experimental Configuration}\label{sec:PHILSetup}
A novel PHIL experimental configuration has been developed to validate the proposed distributed protocol containing 40+ physical hardware generation devices, 48 physical distributed controller nodes (DCNs), 202 simulated DCNs, and a real-time power system model. The components of the experimental configuration (Fig. \ref{fig:PHILsystem}) are described next.
\vspace{-0.3cm}
\subsection{Real-Time Distribution System Model (RT-DSM) Layer}
The underlying RT-DSM used to represent the power network is based on a distribution network model from Aurora, Colorado, USA, which represents around 2000 customers and is augmented in this study to have 50.2\% distributed PV RES penetration (both residential and commercial DER units \textemdash total PV 8.06 MVA, battery ESS (BESS) capacity of 1.203 MW). The network has a peak load of 7.1 MVA and has 5.285 MVA of controllable DERs. 
Irradiance and load profiles from the real-world system were used. The RT-DSM was executed on a 12-core OPAL-RT OP5707 using ePhasorSim at a time step of 10 ms. The capacities of participating DERs are shown in Table \ref{tb:ratingsDER}. PHIL interfaces, including power amplifiers and feedback sensors for measurement, are provided to interconnect the physical hardware under test of Fig. \ref{fig:PHILsystem} at five independent PCCs, $P_1$-$P_5$, in the RT-DSM. 
\vspace{-0.3cm}
\subsection{Simulated Controller (SC) Layer}
The simulated controller layer further consists of: 1) simulated local DER controllers (simLCs) that provide control signals to distributed PV inverters in the RT-DSM. simLC provides real-time measurements of minimum and maximum DER capacity to its DCN interface; 2) simulated distributed controller nodes (simDCNs) formulate the communication interfaces to simLC and provide power reference commands ($P_k^*, Q_k^*,$ see Fig. \ref{fig:PHILsystem}) based on the power apportioning protocol.
\vspace{-0.8cm}
\subsection{Communication Layer}
Rpi-based DCNs, with inherent bounded time delays, are deployed to communicate with phLIS devices and simLCs. A communication topology is generated with the diameter and average node degree corresponding to the distribution network with 250 nodes overlayed on DCNs (simulated and Rpi-based) to represent the weakly-connected neighborhoods of RT-DSM. The WebSocket-based  communication  protocol  between  DCNs was implemented over Ethernet.

\vspace{-0.4cm}
\subsection{Physical Hardware Local Inverter System (phLIS) Layer}
\noindent 1) Residential-scale DER: This 2-kVA custom-built DER (Fig. \ref{fig:HUT} (a)) is a Type-4 Inverter connected to a PV simulator, residential loads, and is interfaced at P1 - 240-V PCC. \\
2) Utility-scale DER (\textit{Type-6 Inverter}): This custom-built DER consists of a BESS (100 kW three-phase, Lithium-ion based 32.8 kWh) and a commercial PV inverter (100 kW, three-phase) (Fig. \ref{fig:HUT} (b)) coordinated via an integrated controller and interfaced at the three-phase PCC P2.\\
3) DER racks: Each DER rack (Fig. \ref{fig:HUT} (c),(d)) consists of two commercial-off-the-shelf residential PV string inverters (Type  1, 3/5 kVA and Type 2 3.8 kVA), a BESS inverter (Type 5, 5 kW) and three racks of 12 PV microinverters (Type 3, 320 W each). Each DER rack is capable of receiving a power reference command via MODBUS protocol, and the racks are interfaced at 240 V PCCs P3 - P5 as shown in Table \ref{tb:ratingsDER}.
Note that the phLIS DERs are scaled through PHIL interfaces to emulate large-scale hardware devices in the RT-DSM. These scalings are presented in Table \ref{tb:ratingsDER}.

\begin{figure}[t]
\vspace{-0.7cm}
    \centering
 \subfloat[]{
            \includegraphics[scale=0.13,trim={2.8cm 0cm 13.5cm 1.9cm},clip]{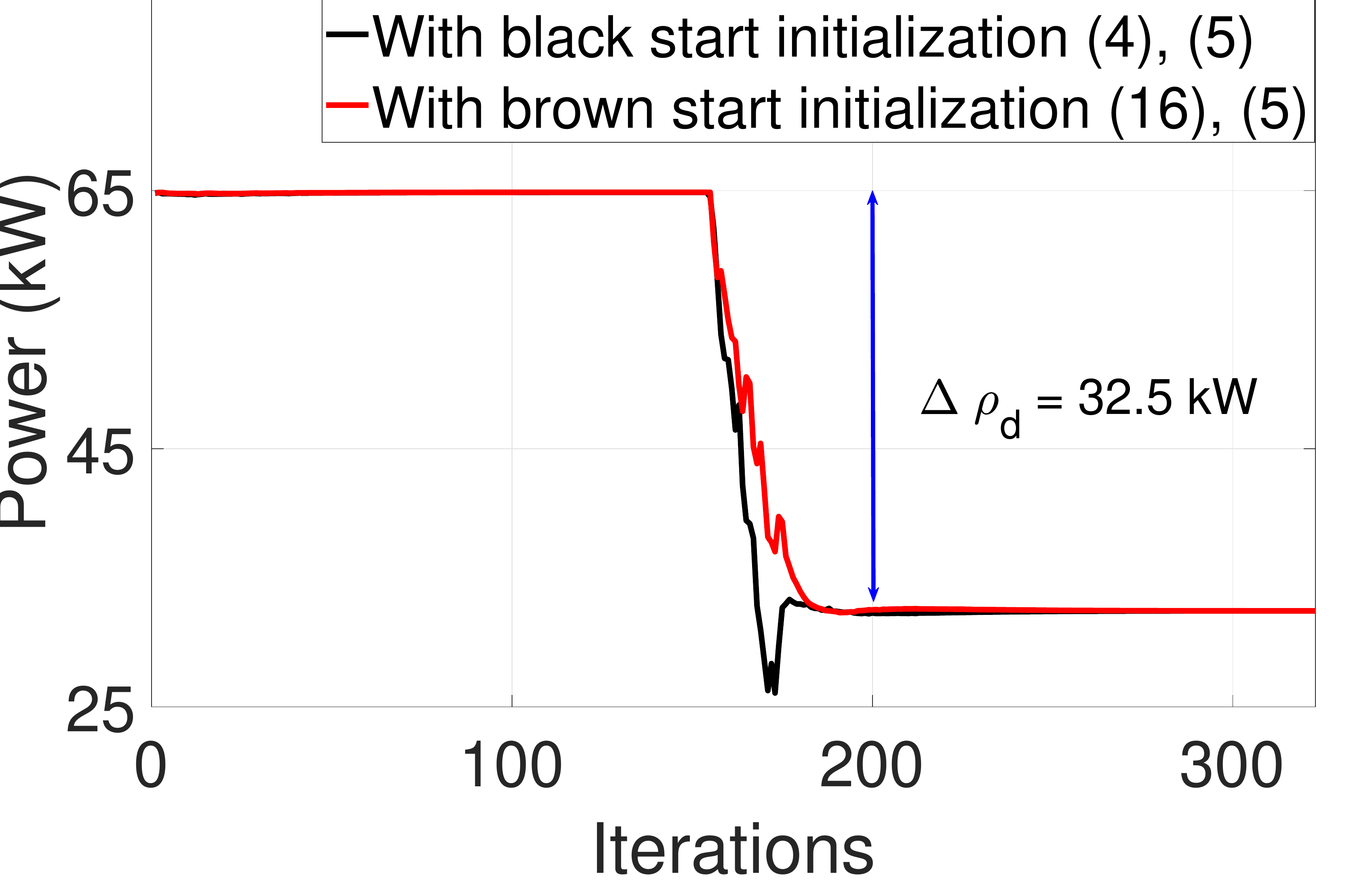}
            }
 \subfloat[]{
            \includegraphics[scale=0.22,trim={0cm 0cm 0cm 0cm},clip]{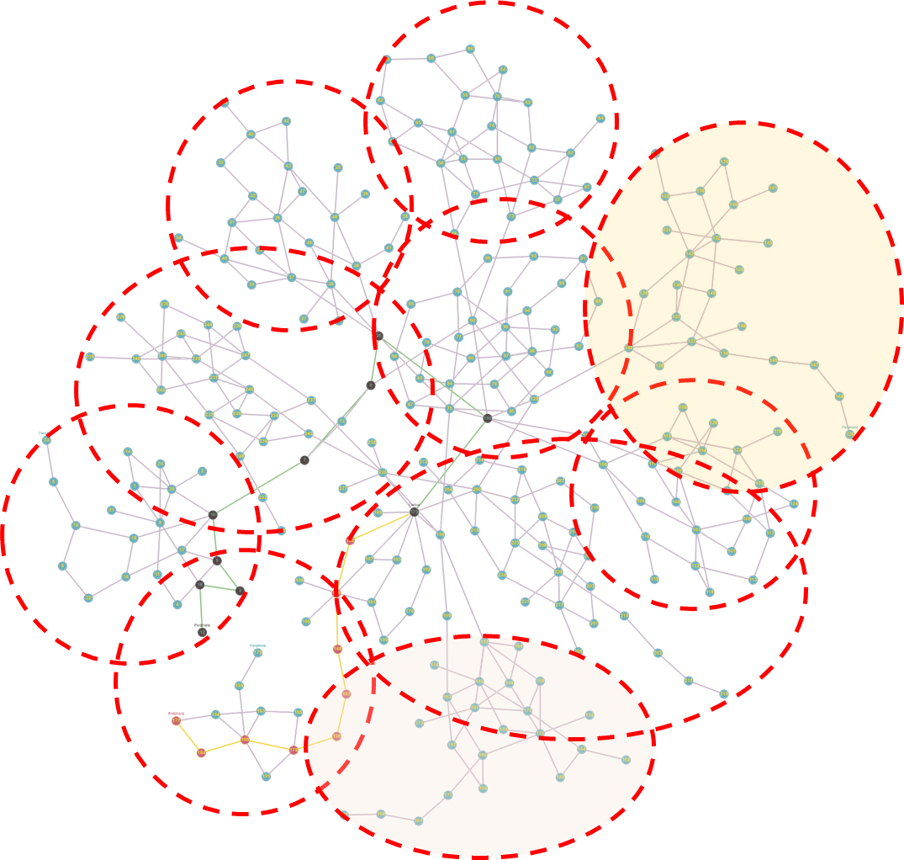}
            }
    
    \caption{(a) Aggregated DER response for decreasing aggregator command with and without brown-start initialization, (b) communication graph ($D=19$) with neighborhoods.}
    \label{fig:ratingsDER}
\end{figure}

\begin{figure}[t]
    \centering
    \includegraphics[scale=0.28]{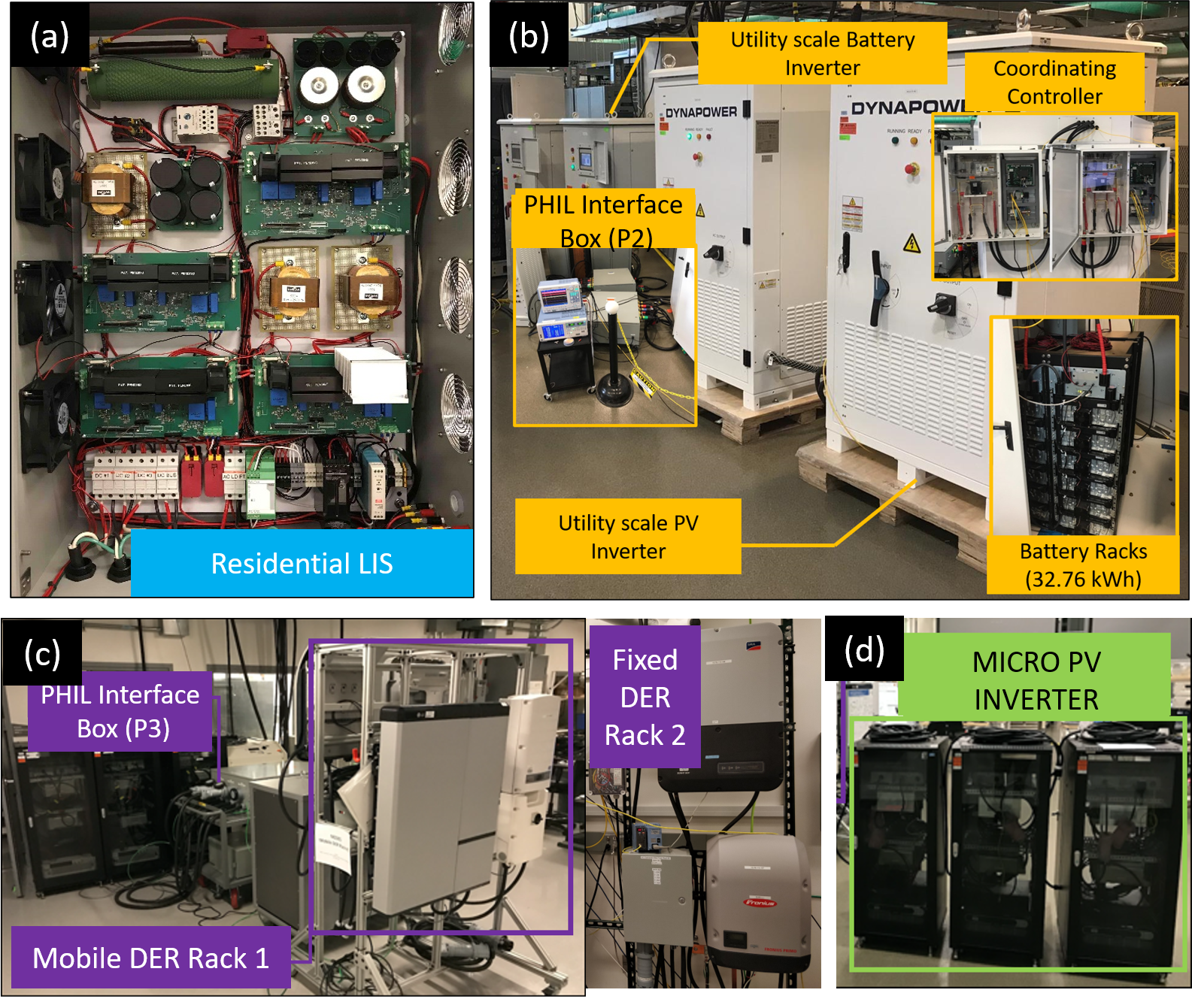}
    \caption{Hardware under test for PHIL experiments.
    }
    \label{fig:HUT}
\end{figure}



\vspace{-0.3cm}
\section{Results}\label{sec:results}
The developed PHIL system is used for demonstrating the capability of the finite-time power apportioning protocol toward providing SFR as an ancillary service to the grid. Each phLIS is facilitated with a RPi-based module for communication with other network DER units. For validation, $\bar{\tau}$ of 50 ms for Test Case I and 20 ms for Test Case II are chosen which are typical round trip times observed. 
The communication topology of Fig. \ref{fig:ratingsDER} (b) is chosen.\\
\vspace{0.3cm}

\begin{table}[b]
\caption{DER RATINGS FOR EXPERIMENTS}
\label{tb:ratingsDER}
\begin{tabular}{ccccccc}
\multicolumn{4}{c}{\textbf{Inverters}}                                                                                                                                                                                                & \multicolumn{2}{c}{\multirow{2}{*}{\begin{tabular}[c]{@{}c@{}}\textbf{Device}\\ \textbf{Count}\end{tabular}}} & \multirow{3}{*}{\begin{tabular}[c]{@{}c@{}}\textbf{PHIL}\\ \textbf{Interface}\\ \textbf{point}\end{tabular}} \\ \cline{1-4}
\textbf{Type }                   & \multicolumn{3}{c}{\begin{tabular}[c]{@{}c@{}}\textbf{Capacity}\\ \textbf{(kW)}\end{tabular}}                                                                                                                        & \multicolumn{2}{c}{}                                                                        &                                                                                   \\ \cline{1-6}
                        & \multicolumn{2}{c}{\begin{tabular}[c]{@{}c@{}}physical\\ \textbf{actual}      \textbf{scaled}\end{tabular}}                                    & sim.                                                                & phys.                                                      & sim.                           &                                                                                   \\ \hline
\multirow{2}{*}{Type 1} & \multirow{2}{*}{3/5/5}                              & \multirow{2}{*}{\begin{tabular}[c]{@{}c@{}}100/75\\ /375\end{tabular}} & \multirow{2}{*}{\begin{tabular}[c]{@{}c@{}}3-6\\ /100\end{tabular}} & \multirow{2}{*}{2/1/1}                                     & \multirow{2}{*}{155/7}         & \multirow{2}{*}{P3-P5}                                                            \\
                        &                                                     &                                                                        &                                                                     &                                                            &                                &                                                                                   \\ \hline
Type 2                  & 3.8                                                 & 375                                                                    & 100                                                                 & 1                                                          & 10                             & P3-P5                                                                             \\ \hline
Type 3                  & \begin{tabular}[c]{@{}c@{}}0.32\\ each\end{tabular} & \begin{tabular}[c]{@{}c@{}}6.25/8.33\\ /0.25\end{tabular}              & -                                                                   & 12/12/12         & -                              & P3-P5                                                                             \\ \hline
Type 4                  & 2                                                   & 2                                                                      & -                                                                   & 1                                                          & -                              & P1                                                                                \\ \hline
Type 5                  & 5/5                                                 & 5/100                                                                  & 7-10                                                                & 1/2                                                        & 65                             & P3-P5                                                                             \\ \hline
Type 6                  & 100                                                 & 1000                                                                   & -                                                                   & 1                                                          & -                              & P2                                                                                \\ \hline
\textbf{Total }                  &                                                     & \begin{tabular}[c]{@{}c@{}}\textbf{2,410} \\ \textbf{(MW)}\end{tabular}                  & \begin{tabular}[c]{@{}c@{}}\textbf{2,864}\\ \textbf{(MW)}\end{tabular}                & \textbf{46 }                                                        & \textbf{237}                            &                                                                                   \\ \hline
\end{tabular}
\end{table}

\vspace{-0.5cm}
\textit{Test Case I: Black Start: DER Aggregation and Early Dispatch with Fixed Capacities}\\

We consider the scenario where the SFR commanded by the aggregator is $\rho_d$ (instead of $\Delta \rho_d$) and participating DERs have fixed capacities. The RT-DSM follows a net-load profile where the feeder consumption supplied by the bulk power system is 200 kW (Fig. \ref{fig:subPower_noED} (a)) and DER units collectively provide 225 kW. A fixed irradiance of $800~ W/m^2$ is considered for Test Case I. Distributed PV inverters in the network have headroom available to provide SFR.
At $t=27$ s, the aggregator sends a new total power command, $\rho_d = 500\, kW$,  to achieve a change in feeder net active power consumption by $\Delta \rho_d = -275 \, kW$ (see Fig. \ref{fig:subPower_noED} (b). Two command circulating nodes receive the command signals, $\rho_d/2$, each and initiate the power apportioning algorithm with early dispatch. $\rho$ is set to 0.01 for all test cases. Fig. \ref{fig:DER1power_noED} shows the reference power command updates of the three Rpi-DCNs associated with Mobile DER Rack 1 and corresponding DER output power. Clearly, the initial response of the DERs is obtained within 2 s of receiving the aggregator command by the command circulating nodes, and the finite-time criteria is met within 50 s for the distribution network under validation. Fig. \ref{fig:DER_noED} shows the power reference updates and the total active power reference for all the 250 participating DERs. These results establish that the desired SFR command is met and the response times achieved by the simulated and phLIS DERs satisfy the initial response time requirement of 5 seconds and ramp time to response of 1 minute.
Figure \ref{fig:subPower_noED}(a), compares the case with and without DER' dispatch to meet active power in the feeder. The response of the feeder with respect to the increase in aggregated DER' output is presented in Fig. \ref{fig:subPower_noED}(b). The red line here corresponds to the change in dispatch command received by $\mathcal{N}_d$. The secondary feeder voltages were  remain well within $\pm 1\%$ of nominal (ANSI C84.1).\\


\begin{figure}[t]
    \centering
    \subfloat[]{\includegraphics[scale=0.21,trim={0.35cm 0cm 0.75cm 0.5cm},clip]{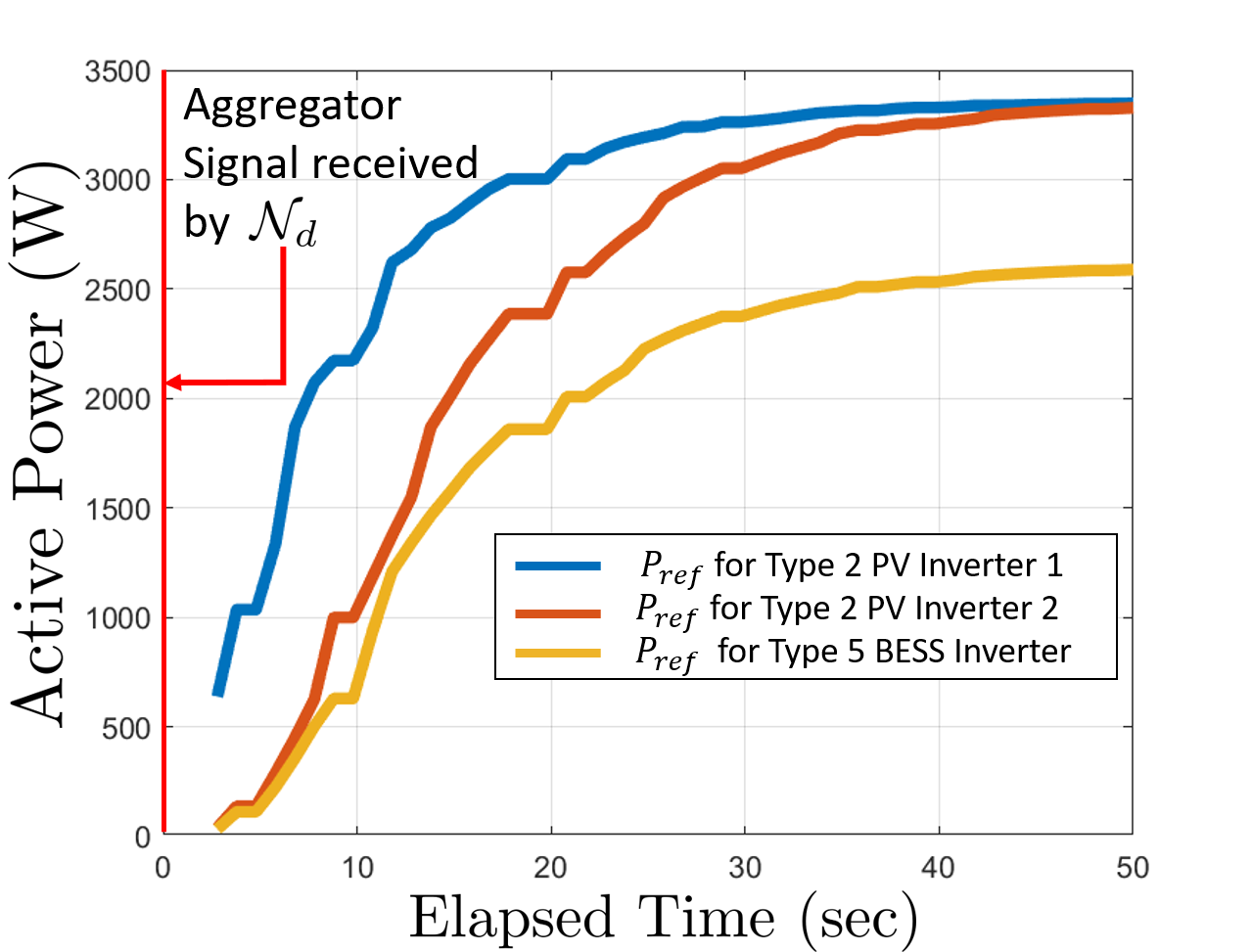}}
    \subfloat[]{\includegraphics[scale=0.21,trim={0.35cm 0cm 0.75cm 0.5cm},clip]{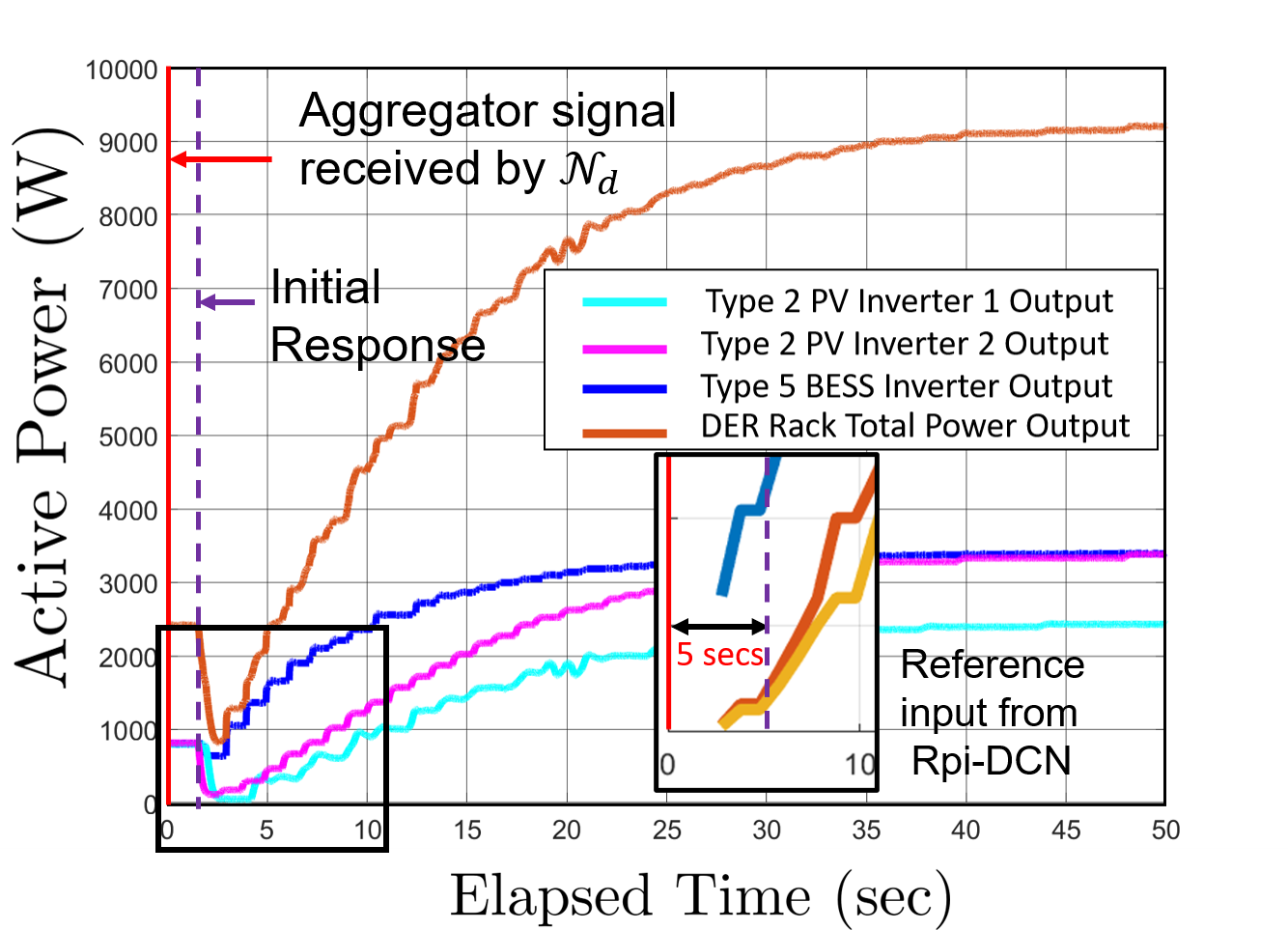}}
    \caption{(a) Power reference commands at Rpi-DCNs, (b) dispatched output power for DER Rack 1 inverters. }
    \label{fig:DER1power_noED}
\end{figure}

\begin{figure}[t]
    \centering
    \subfloat[]{\includegraphics[scale=0.2,trim={0.35cm 0cm 0cm 0cm},clip]{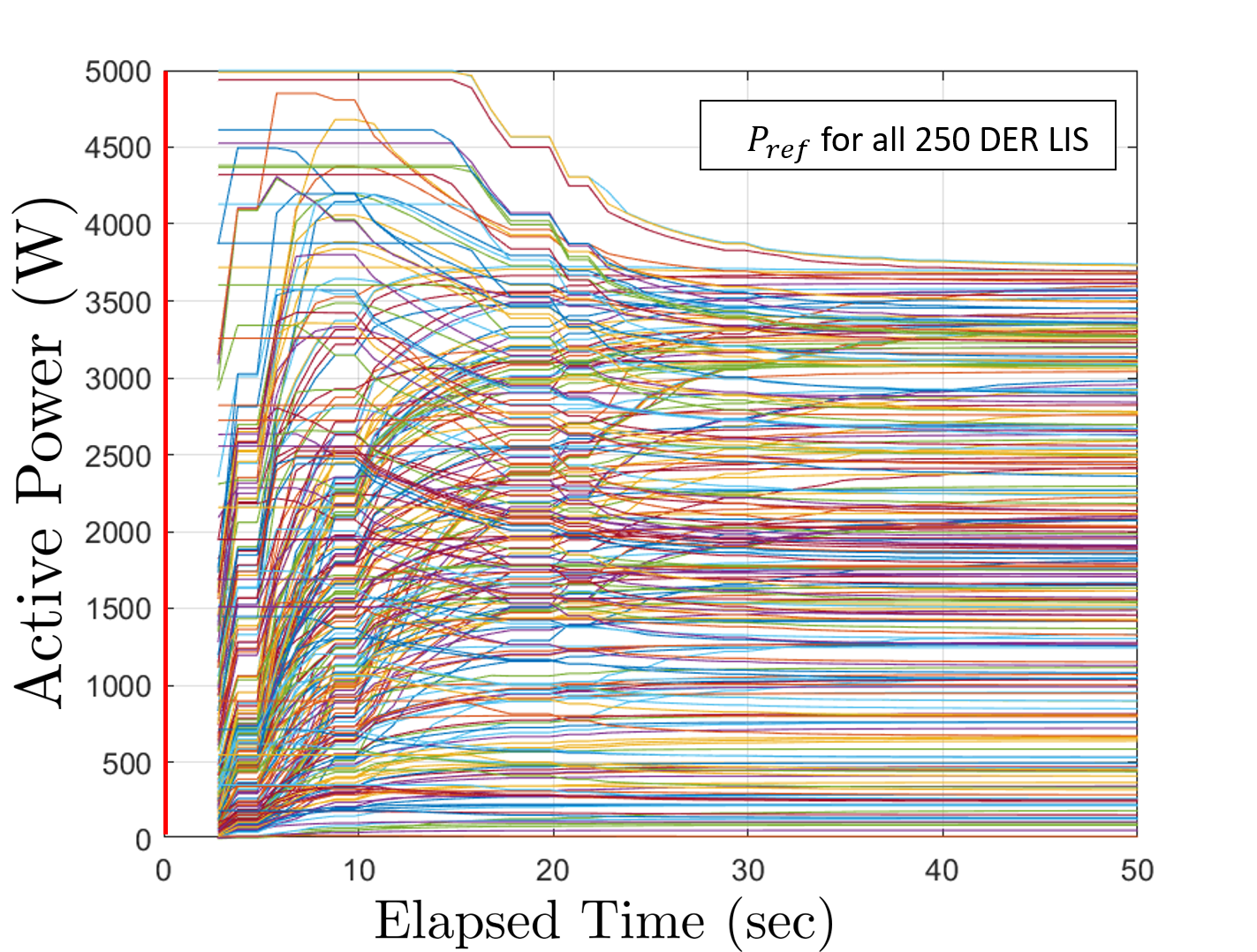}}
    \subfloat[]{\includegraphics[scale=0.21,trim={0.73cm 0cm 0cm 0cm},clip]{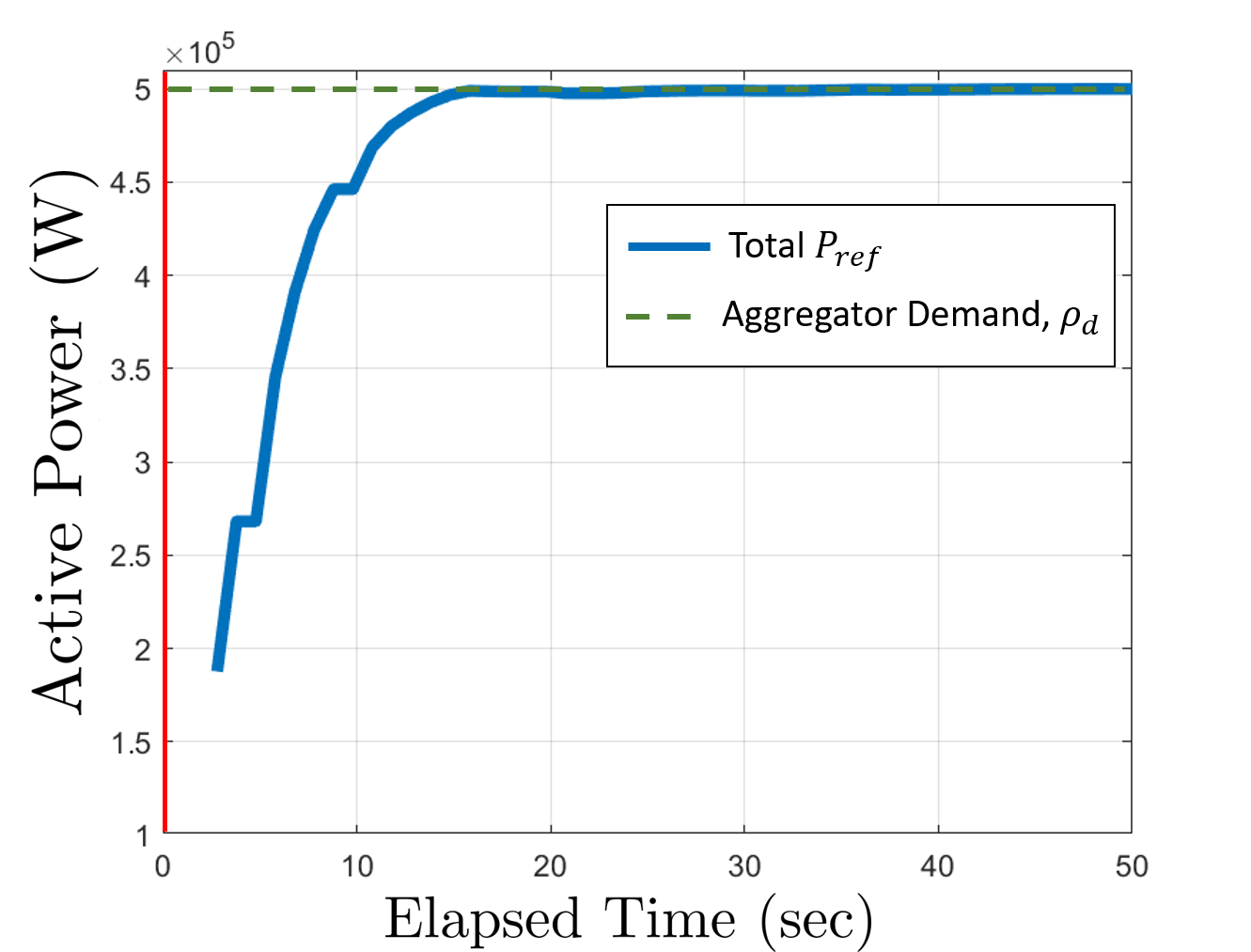}}
    \caption{(a) Reference power command updates for 250 DER units in the network, (b) Total reference power command from the DER network with $\Delta \rho_d = 275 \, kW$ and $\rho_d = 500$ kW.}
    \label{fig:DER_noED}
\end{figure}
\begin{figure}[t]
    \centering
    \includegraphics[scale=0.23]{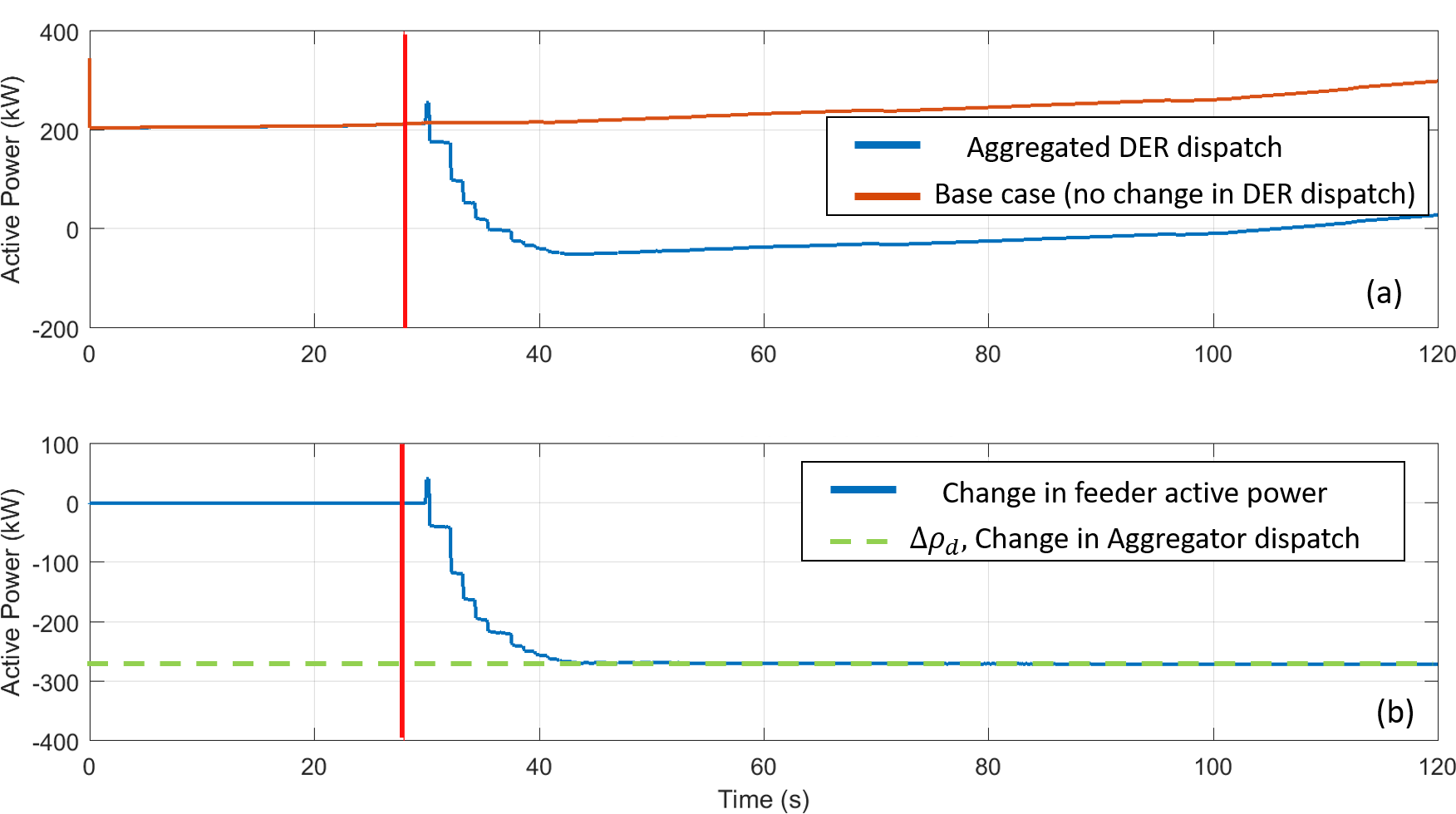}
    \caption{Active power consumed by distribution network with and without power apportioning protocol for Test Case-I.}
    \label{fig:subPower_noED}
    \end{figure}
\textit{Test Case II: Brown Start: DER Aggregation with Early Dispatch for Time-Varying Capacities and Aggregator command}\\

\begin{figure}[t]
    \centering
    \includegraphics[scale=0.2,trim={3.8cm 0cm 0cm 0cm},clip]{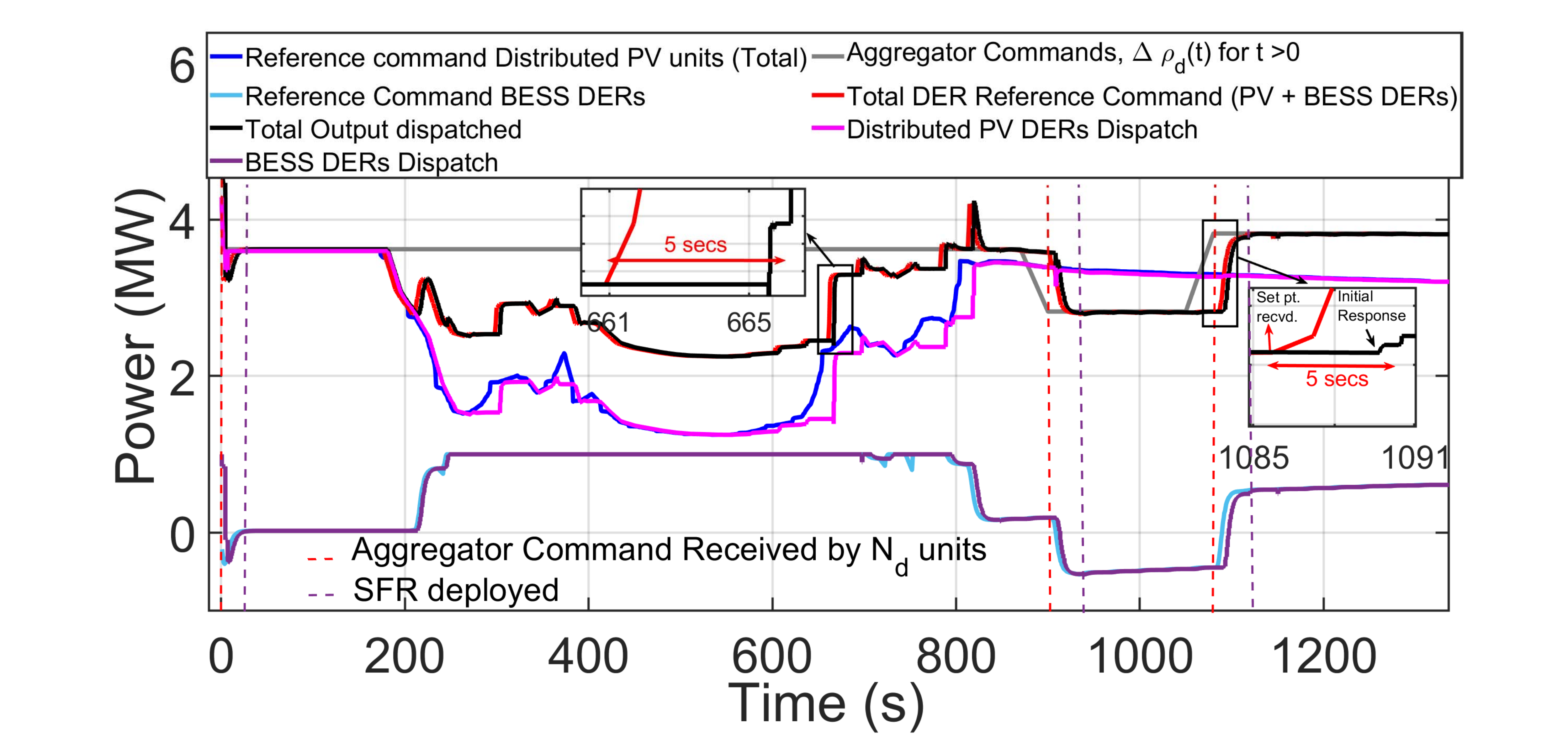}
    \caption{Reference power commands obtained at DCNs and output power dispatched for Test Case II.}
    \label{fig:PHIL_TC2}
\end{figure}

\begin{figure}[t]
    \centering
     \includegraphics[scale=0.18, trim={1.8cm 0cm 2cm 0cm},clip]{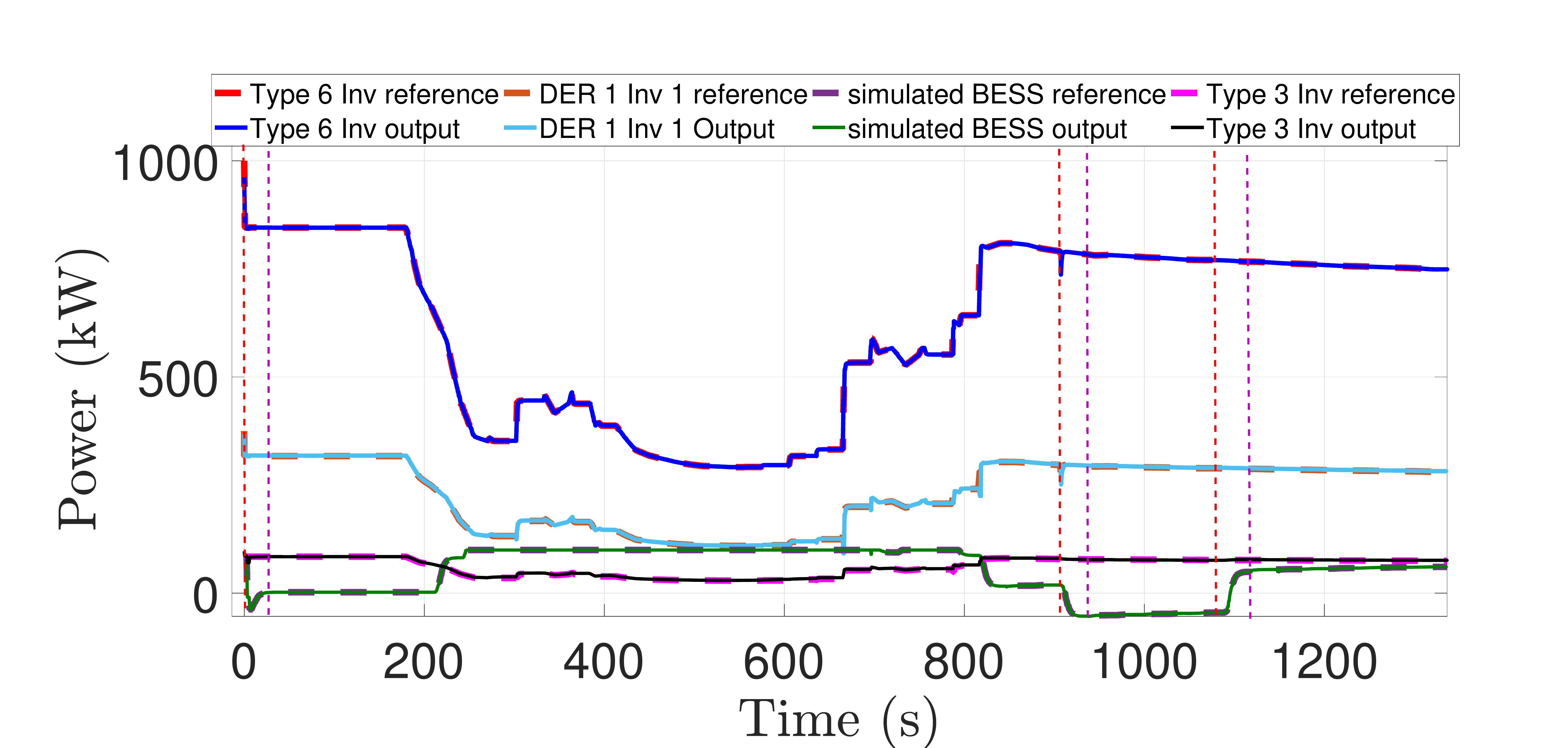}
    \caption{Reference Power command from Rpi-DCNs/sim-DCNs and scaled dispatched power from (selected) phLIS.}
    \label{fig:TC2_individualDER}
\end{figure}
 This test case emphasizes SFR due to continuous dispatches and variability in PV generation capability, requiring adjustments of minimum and maximum capacities of network DERs. In this scenario, the brown-start mechanism is implemented to achieve power apportioning due to changes in the system operator's command and/or changes in generation capacities of participating PV and BESS DERs. Here, each PV inverter follows an irradiance profile sampled every 30 seconds.
 Initially, at t = 0 s, $\mathcal{N}_d = 2$ (here Mobile DER Rack 1: Type 1 and Type 2 inverters) receive the initial aggregator demand $\rho_d(t_0) = 3.62~MW $ and within 5 s all the units in the network respond initially to meet the net command. A $\Delta \rho_d(t)$ of -0.82 MW at $t=900$ s and +1.01 MW at $t = 1080$ s are provided as aggregator commands to the network via $\mathcal{N}_d$ (see Fig. \ref{fig:PHIL_TC2}). 
At time periods,  $0 < t < 900$ s, $ 900<t< 1080 $ s, and $ t > 1080$ s, the brown-start and early dispatch mechanisms allow DERs to be dispatched to achieve feasible solutions to (\ref{eq:resApportioning}) in the presence of  variable solar irradiance and discharging/charging BESS DERs. The ramp response time (within $\pm 5\%$ of the steady-state value) to meet these commands by the aggregate DER network in the RT-DSM is $\sim 21 - 44$ s (to within SFR dispatch timescales) as shown by the purple dotted lines in Fig. \ref{fig:PHIL_TC2}.  Fig. \ref{fig:TC2_individualDER} shows the commanded power from physical DERs and simulated DERs in RT-DSM as well as the output response of selected physical hardware and simulated devices in the system. From $0 \le t \le 180$ s, because the total demand can be solely met by DERs with RES prioritization, due to the high irradiance that allows the distributed PV units to operate at their full capacity, the BESS units are not dispatched. For $180 < t < 800$ s, however, because of the decrease in irradiance, the net output power of all distributed PV units decreases, causing the algorithm to command BESS units in the network to dispatch their cumulative maximum capacity of ~1200 kW in order to minimize the difference in commanded and generated power. Similarly, at $t = 800$ s, when RES generation output was higher, BESS dispatch was adjusted back to meet the deficit. At $t = 900$ s, because the RES generation is higher than demanded $\rho_d(t)$, surplus PV generation is used to charge the BESS units. Finally, at $t = 1080$ s, both RES and BESS units are dispatched to meet and sustain the changed power commands of the aggregator with a ramp response time of less than 50 s. 
Finally, the secondary-side feeder voltages were observed to have voltage deviations to be within $ \pm 5\%$.

\section{Conclusion}\label{sec:conclusion}
This article develops a scalable distributed framework for coordinating and aggregating large numbers of DERs to provide ancillary service support in the form of SFR to a bulk power system. The results of finite-time termination of ratio consensus were extended to propose a distributed power apportioning protocol to meet the time-varying aggregator command within a specified tolerance by participating DERs in the presence of bounded communication delays while prioritizing RES in the network. DER responses faster than state-of-the-art distributed approaches required for SFR services were achieved by implementing an early dispatch mechanism along with a brown-start approach to further improve DERs' ramping performance.
Experimental results, validated on a unique PHIL testbed at scale, show that the required performance metrics were met (initial and ramp responses were achieved to be  less than $5$ seconds and $< 50$ seconds for the setup, respectively). 



\vspace{-0.3cm}

\bibliographystyle{IEEEtran}
\bibliography{topident,biblio}

\appendices
\ifCLASSOPTIONcaptionsoff
  \newpage
\fi


%








\end{document}